\documentclass[preprint,12pt]{elsarticle}
\usepackage{amssymb}
\usepackage{amsthm}
\usepackage[utf8]{inputenc} 
\usepackage[T1]{fontenc}    

\journal{Games and Economic Behavior}
\journal{arXiv.}

\usepackage[hyperfootnotes=false,colorlinks=true,citecolor=blue]{hyperref}

\usepackage{url} 

\usepackage{multicol,lipsum}
\usepackage{fancyvrb}

\usepackage{amsthm}
\usepackage{amssymb}
\usepackage{amsmath}
\usepackage{mdframed}
\usepackage{enumitem}
\usepackage{xspace}
\usepackage{sgame}
\usepackage{fmtcount} 
\usepackage{thmtools}
\usepackage{mleftright}

\usepackage{natbib}

\newcommand{\Example}{Example\xspace}

\newcommand{\examples}{examples\xspace}

\declaretheoremstyle[
spaceabove=2.5ex, 
spacebelow=0ex,
numberwithin=section
]{styleone}
\declaretheorem[style=styleone]{theorem}
\declaretheorem[style=styleone]{proposition}

\newtheorem{lemma}[theorem]{Lemma}

\declaretheoremstyle[
spaceabove=2.5expt, 
spacebelow=0ex,
headfont=\normalfont\bfseries,
notefont=\mdseries, 
notebraces={(}{)},
bodyfont=\normalfont,
postheadspace=1em,
qed=\qedsymbol
]{definitionstyle}
\theoremstyle{definition}
\newtheorem{definition}[theorem]{Definition}
\newtheorem{notation}[theorem]{Notation}
\newtheorem{exm}[theorem]{\Example}

\newmdtheoremenv[innertopmargin=0pt,skipabove=1ex]{problem}{Open Problem}

\newcommand{\qedblack}{\hfill\ensuremath{\blacksquare}}
\newenvironment{challenging}{
\textit{Why this problem is challenging.\xspace}
}{
\hfill \qedblack
}

\usepackage{todonotes}
\usepackage{soul}
\usepackage{marginnote} %

\newif\ifcomments 

\newif\iffootnote
\let\Footnote\footnote
\renewcommand\footnote[1]{\begingroup\footnotetrue\Footnote{#1}\endgroup}

\newcommand\setupcomments[4]{
\expandafter\def\csname#1says\endcsname##1{\ifcomments{\color{#3}[#2: ##1]}\fi}
\expandafter\def\csname#1predicts\endcsname##1{\ifcomments{\color{#3}[#2 predicts: ##1 \#}\fi}
\expandafter\def\csname#1color\endcsname##1{\ifcomments{\color{#3}##1}\fi}
\expandafter\def\csname#1adds\endcsname##1{\ifcomments{\color{#3}##1}\fi}
\expandafter\def\csname#1rems\endcsname##1{\ifcomments{\color{#3}\st{##1}}\fi}
  \expandafter\def\csname#1box\endcsname##1{\ifcomments\vspace{1ex}\todo[inline,#4]{#2 says: ##1}{}\fi}

  \expandafter\def\csname#1margin\endcsname##1{\ifcomments\marginnote{\todo[inline,#4]{#2 says: ##1}{}}\fi}
  \expandafter\def\csname#1section\endcsname##1{\ifcomments\todo[#4]{#2 says: ##1}{}\fi}
  \expandafter\def\csname#1todo\endcsname##1{\ifcomments\marginnote{\todo[inline]{#2 says: TODO: ##1}{}}\fi}

  \expandafter\def\csname#1demo\endcsname{\ifcomments
\noindent
$\backslash$#1says $\to$ \csname#1says\endcsname{example inline comment from #2.}\\
$\backslash$#1adds $\to$ \csname#1adds\endcsname{example addition by #2.}\\
$\backslash$#1rems $\to$ \csname#1rems\endcsname{example removal by #2.}\\
$\backslash$#1box $\to$ 
    \csname#1box\endcsname{example box comment from #2}
\fi}

}

\commentstrue
\commentstrue
\setupcomments{ac}{AC}{green!70!black}{color=green!10}
\setupcomments{sr}{SR}{red!80!black}{color=red!30}

\definecolor{mintedgreen}{RGB}{8,127,1}
\definecolor{mintedblue}{RGB}{5,2,255}
\definecolor{mintedred}{RGB}{186,33,33}

\newcommand{\easy}{{\color{mintedgreen}[easy]}\hspace{1ex}}
\newcommand{\medium}{{\color{mintedblue}[medium]}\hspace{1ex}}
\newcommand{\difficult}{{\color{mintedred}[difficult]}\hspace{1ex}}

\newcommand{\answer}[1]{(answered in footnote\footnote{#1})}

\newcommand{\proves}[1]{\underset{#1}{\vdash}}

\newcommand{\bxof}[2]{\Box_{#1}\mleft(#2\mright)}
\newcommand{\NN}{\mathbb{N}}
\newcommand{\Oo}{\mathcal{O}}

\makeatletter
\newcommand{\Newlist}[4]{
 #1#4\@ifnextchar\bgroup{#2\Newlist{}{#2}{#3}}{#3}
}
\makeatother

\newcommand{\quants}[2]{{\Newlist{\mleft(}{\mright)\mleft(}{\mright)} #1}\mleft(#2\mright)}

\newcommand{\Forall}[2]{\quants{{\forall #1}}{#2}}

\newcommand{\Implies}[1]{\Rightarrow}
\newcommand{\floor}[1]{\lfloor #1 \rfloor}

\newcommand{\VC}{\Verb{C}\xspace}
\newcommand{\VD}{\Verb{D}\xspace}
\newcommand{\VCC}{\Verb{(C,C)}\xspace}
\newcommand{\VCD}{\Verb{(C,D)}\xspace}
\newcommand{\VDC}{\Verb{(D,C)}\xspace}
\newcommand{\VDD}{\Verb{(D,D)}\xspace}

\newcommand{\lob}{L\"{o}b\xspace}
\newcommand{\godel}{G\"{o}del\xspace}

\makeatletter
\newcommand{\auxdef}[2]{
  \protected@write\@auxout{}{%
    \global\string\@namedef{#1}{#2}%
    }%
  }

\newcommand{\auxref}[1]{
    \@ifundefined{#1}{??}{\@nameuse{#1}}%
    }
\makeatother

\usepackage[frozencache=true,cachedir=minted-cache]{minted}
\input{./minted-cache/default.pygstyle}

\begin{document}

\begin{frontmatter}



\title{Cooperative and uncooperative institution designs: \\ Surprises and problems in open-source game theory}



\author[inst1]{Andrew Critch} 
\ead{critch@eecs.berkeley.edu}


\author[inst1]{Michael Dennis}
\author[inst1]{Stuart Russell}
\address{Center for Human-Compatible AI, University of California, Berkeley}

\begin{abstract}
It is increasingly possible for real-world agents, such as software-based agents or human institutions, to view the internal programming of other such agents that they interact with.  For instance, a company can read the bylaws of another company, or one software system can read the source code of another.  Game-theoretic equilibria between the designers of such agents are called \emph{program equilibria}, and we call this area \emph{open-source game theory}.

\vspace{1em}
\noindent In this work we demonstrate a series of counterintuitive results on open-source games, which are independent of the programming language in which agents are written.  We show that certain formal institution designs that one might expect to defect against each other will instead turn out to cooperate, or conversely, cooperate when one might expect them to defect.  The results hold in a setting where each institution has full visibility into the other institution's true operating procedures.  We also exhibit \examples and \auxref{probcount} open problems for better understanding these phenomena.
We argue that contemporary game theory remains ill-equipped to study program equilibria, given that even the \emph{outcomes} of single games in open-source settings remain counterintuitive and poorly understood.  Nonetheless, some of these open-source agents exhibit desirable characteristics---e.g., they can unexploitably create incentives for cooperation and legibility from other agents---such that analyzing them could yield considerable benefits.
\end{abstract}




\end{frontmatter}



\section{Introduction}
Numerous sectors of the global economy are positioned for rapid advancements in automation, due to recent progress in artificial intelligence and machine learning \citep{dirican2015impacts,makridakis2017forthcoming,huang2018artificial,wirtz2018brave,hatcher2018survey}.  In fact, many experts speculate that AI technology will eventually outperform humans on a very wide variety of tasks \citep{muller2016future,brynjolfsson2017business,grace2018will}.  
This potentiality means it is important to understand how automated decision-making systems can interact, especially in cases where their behavior might come as a surprise to their designers, such as in the stock market ``flash crash'' of 2010 \citep{madhavan2012exchange}.  

In addition, opportunities abound in the structuring of novel human institutions both for using and for governing AI technology.  There is already a great deal of academic interest in establishing policy and governance norms around the development and deployment of AI \citep{calo2017artificial,gasser2017layered,dafoe2018ai,dwivedi2019artificial,cihon2019standards}.   \citet{dafoe2020open} specifically call for further attention to the risks and benefits of designing AI systems to cooperate with each other.

What body of game-theoretic research and intuitions should be informing the deliberation of policy-makers in the AI governance arena, or of AI researchers in the development and deployment of software-based agents?  Numerous results from classical game theory, bargaining, and mechanism design will no doubt prove useful.  However, as we shall see, novel possibilities arise in strategic interactions between algorithms, some of which are counterintuitive to humans, and some of which might shed new light on existing institutional behaviors or provide new routes to advantageous outcomes. 

In this paper, we demonstrate a series of counterintuitive results showing that certain cooperative-seeming formal institution designs turn out to defect against each other---or conversely, turn out to cooperate when seeming like they would defect---in settings where each institution's true operating polices are visible to the other institution.  We formalize each institution as an \emph{open-source agent}, i.e., a program whose \emph{source code} is readable by other programs.  Source code provides a written specification of how a program will behave; it's what a software developer writes to create a new application, and can be seen as a formal model of the processes by which an institution will operate.

A natural starting point for analyzing open-source agents is the setting of \emph{program equilibria} introduced by \citet{tennenholtz2004program}.  In this setting, given a game $G$ to be played by two agents (``programs'' in the terminology of \citeauthor{tennenholtz2004program}), we consider a higher-level game $G'$ being played by the designers of those agents.  In $G'$, each designer chooses an agent to play in the game $G$, and each agent reads the other agent's source code (as a string of text) before choosing its action (or `strategy') for $G$.  Mathematically, a similar formalism is employed by \citet{kalai2010commitment}, where again, each agent (`device response function' in their terminology) gets to read the other agent's source code (`commitment device') before choosing its action (`strategy').  We call this area \emph{open-source game theory}. 

While the frameworks of Tennenholtz and Kalai et al.~differ on when exactly the designers are allowed to introduce randomization, in both settings the agents are open-source. 
Tennenholtz mainly takes the perspective that the open-source condition will be necessary for modelling the economic interaction of real-world software systems via the Internet.  
This perspective could become increasingly important as artificial intelligence progresses, or if so-called ``decentralized autonomous organizations'' (DAOs) ever become economically prominent \citep{jentzsch2016decentralized,chohan2017decentralized,dupont2017experiments}.  
On the other hand,  \citeauthor{kalai2010commitment} view their formalism mainly as a means of representing commitments between human institutions, such as in price competitions or legal contracts. In our view, both potential applications of the open-source agent concept---to artificial agents and to human institutions---are important.  Another shortcoming of Tennenholtz' approach is that he invents a particular programming language in which to write agents, which is not immediately applicable to agents written in other programming languages or formalisms.

In this paper, we prove theorems about open-source agents \emph{that do not depend on which programming language they're written in}, using the formalism of mathematical proofs as applicable to arbitrary Turing machines.  We argue that the consequences of the open-source condition are, in short, more significant than previously made apparent in the game theory literature.  In particular, even the \emph{outcome} of a single game between fully specified open-source agents can be quite counterintuitive, as the theorems of this paper will demonstrate.  It is possible, therefore, that the overlap between computer science and game theory, with a helping of logic, could progress into what von Neumann, E. Kalai, and  Shoham would call ``the third stage'' of scientific development, where theory is able to make accurate predictions beyond what practitioners would intuitively expect \citep{shoham2008computer}.  
Intriguingly, such results may be \emph{necessary} for historically earlier concepts from game theory, such as equilibria, to be applicable in reality: without these results, the designers of artificial agents cannot acquire an adequate understanding of those agents to be able to foresee even a single game outcome, so it would make little sense to model those designers as being in equilibrium with one another.

For this reason, we call for the study of open-source games in their own right, irrespective of whether the designers of the agents in the games are in equilibrium.  Thus, the study of open-source game theory as construed here is both broader and more fundamental than that of program equilibria.  We begin this study by examining interactions between very simple agents who reason about each other using mathematical proofs and make decisions based on that reasoning (Section \ref{sec:setup}).  We show how results in mathematical logic can be applied to resolve the outcomes of such games (Section \ref{sec:results}) and conjecture that perhaps these results may be reducible to a form more easily applied to human institutions (Section \ref{sec:intuition}).

We anticipate that the methods of this paper will be unfamiliar to many readers in game theory and economics.  Facing this difficulty may be inevitable if we wish to understand the game theory of artificial agents or to learn from that theory in our modelling of human institutions.  To encourage readers to think about this 
 topic through fresh eyes, we include simple workable examples alongside theorems and open problems, with the hope of inspiring more interdisciplinary work spanning game theory, computer science, and logic.

\subsection{Related work}\label{sec:rw}
As discussed in the introduction, \citet{tennenholtz2004program} and \citet{kalai2010commitment} have previously examined the interaction of open-source agents and characterized the equilibria between the designers of such agents in the form of so-called ``folk theorems.''   \citeauthor{tennenholtz2004program} calls the open-source agents ``programs'' and \citeauthor{kalai2010commitment} call them ``device response functions.''  \citeauthor{tennenholtz2004program}, as we sometimes do, allows the agents to randomize at runtime, and does not examine settings where the designers can randomize.
By contrast, the equilibrium concept of \citeauthor{kalai2010commitment} assumes the agents are deterministic at runtime whereas the designers are allowed to randomize in choosing their agents.  
Kalai's equilibrium concept also allows the designers to choose whether their agents are open-source.  In this paper, we focus away from designer equilibria and toward understanding the details of what kinds of agents can be built and how they interact.

\citet{tennenholtz2004program} also demonstrates a fairly trivial means of achieving a mutually cooperative equilibrium between agent designers: by writing an agent that checks if its opponent is exactly equal to itself, and cooperates only in that case.  \citet{lavictoire2014program} remark that this cooperative criterion is too fragile for practical use and put forward the idea that formal logic, particularly L\"{o}b's theorem and more generally \godel-\lob provability logic, should play a role in agents thinking about each other before making decisions. 
For the case in which the agents are given infinite time and space in which to do their thinking, \citet{barasz2014robust} exhibit a finite-time algorithm for ascertaining the outcomes of games between these (necessarily physically impossible) infinitely resourced agents, called \emph{modal agents} because of their use of modal logic.   
Currently it remains unclear to the present authors whether existing results in logic are truly adequate to establish bounded versions of \emph{all} of the modal agents described by \citet{barasz2014robust} and \citet{lavictoire2014program}; it is, however, a highly promising line of investigation.

\citet{oesterheld2019robust} exhibits a mutual simulation approach to cooperation in an open-source setting and argues that this approach is more computationally efficient than formally verifying properties of the opponent's program using proofs, as we do in this paper.  However, as Section \ref{sec:efficiency} will elaborate, mutual program verification can be made more efficient than mutual simulation, by designing the verification strategy to prioritize hypotheses with the potential to collapse certain loops in the metacognition of the agents.

\citet{halpern2018game} and  \citet{capraro2019translucent} have shown how partial transparency, or ``translucency'', between players is more representative of real-world interactions between human institutions and creates more opportunity for pro-social outcomes.  Their results assume a state of common knowledge in which the agents are known to be rational and counterfactually rational, which may be predictably unrealistic for present-day systems and institutions.  However, the normative appeal of their results could be used to motivate the design of more translucent and counterfactually rational entities.

Loosely speaking, \emph{revision games} \citep{kamada2020revision} and \emph{mechanism games} \citep{yamashita2010mechanism,peters2013folk,peters2014competing} are settings in which some players make plans or commitments visible to others.  Since plans and sometimes commitments can be written as programs, results from studying interactions between open-source agents might reveal interesting new strategies for revision games and mechanism games.

\section{Setup}\label{sec:setup}

In this paper, agents will make decisions in part by formally verifying certain properties of each other, i.e., by generating mathematical proofs about each other's source code.  We'll call such agents \emph{formal verifier agents}.  

For ease of exposition, pseudocode for the agents will be written in Python because it is a widely known programming language.  However, we emphasize that unlike \citet{tennenholtz2004program}, we \textit{are not} inventing a programming language for representing agents.  Rather, we analyze agents (and describe agents that analyze each other) using a proof-based approach, because proofs can be written that are independent of the particular programming language used to create the agents.  The agents could even be written in different programming languages altogether, as long as they are Turing machines.  Still, empirical or exploratory research in this area will probably be more efficient if the agents are written in a programming language more specifically designed for formal verifiability, such as HOL/ML \citep{nipkow2002isabelle,klein2009sel4} or Coq \citep{barras1997coq,chlipala2013certified}, but we emphasize that our results do not depend on this.

Below is a simple agent, called \Verb{CB} for ``CooperateBot'', who simply ignores its opponent's source code \Verb{opp_source} and returns \VC for ``cooperate''.  ``DefectBot'' is the opposite, returning \VD for ``defect''.  For any agent $A$, $A$\Verb{.source} will carry a valid copy of the agent's source code, like this:


\begin{minted}{python}
# "CooperateBot":
def CB(opp_source):
    return C
CB.source = """
def CB(opp_source):
    return C
"""
\end{minted}

\begin{minted}{python}
# "DefectBot":
def DB(opp_source):
    return D
DB.source = """
def DB(opp_source):
    return D
"""
\end{minted}


To save space when defining subsequent agents, we'll avoid writing out the full \Verb{agent.source} definition in this document, since it's always just a copy of the lines directly above it.

The outcome of a game is defined by providing each agent's source code as input to the other:
\begin{minted}{python}
# Game outcomes:
def outcome(agent1, agent2):
    return (agent1(agent2.source),agent2(agent1.source))
\end{minted}

To get warmed up to thinking about these agents, and to check that our definitions are being conveyed as intended, we ask the reader to work through each of the short \examples posed throughout this paper.  Later \examples will turn out to be open problems, but we'll start with easier ones:

\vspace{1ex}
\begin{exm} \easy What is \Verb{outcome(CB,DB)}? \answer{\Verb{CB(DB)} returns \VC because \Verb{CB} always returns \VC by definition; similarly  \Verb{DB(CB)} returns \VD. Hence the outcome is the pair \Verb{(C,D)}.}
\end{exm}

\subsection{Proof-searching via proof-checking}\label{ssec:proofsearch}

In this paper, our goal is to look at agents who perform some kind of check on the opponent before deciding whether to cooperate or defect.  For this, let's suppose the agents are equipped with the ability to read and write formal proofs about each other's source code using a formal proof language, such as Peano Arithmetic or an extension thereof.  Peano Arithmetic and its extensions are useful because they allow proofs about programs \emph{written in arbitrary programming languages}, by representing arbitrary computable functions in a mathematical form \citep{cori2001mathematical}.  In this way, we move past what might have seemed like a limitation of previous works on program equilibria, which invented and employed domain-specific programming languages for representing agents \citep{tennenholtz2004program, lavictoire2014program}.\footnote{The perception that program equilibrium results require agents to be written in domain-specific languages work was remarked by several earlier reviewers of the present draft.  The issue of whether it's ``obvious'' that the programming language requirements can be broadly generalized seems to depend starkly on the reader.  In any case, in light of the programming-language-independent results presented here and by \citet{critch2019parametric}, earlier works \citep{tennenholtz2004program, lavictoire2014program} can be argued to be more generalizable than they may have seemed at the time.}  Specifically, we will assume the agents can invoke a function called \Verb{proof_checker} which, given a source code for an opponent agent (encoded as a string), can check whether a given argument (encoded as a string, referring to the opponent as ``opp'') is a valid proof of a given proposition (also encoded as a string) about the opponent:
\begin{minted}{python}
def proof_checker(opp_source, proof_string, hypothesis_string):
 # This is the only function in this paper that won't be
 # be written out explicitly.  Our assumption is that it
 # checks if proof_string encodes a logically valid
 # proof that hypothesis_string is a true statement about
 # the opponent defined by opp_source.  The hypothesis_string  
 # can refer to the opponent as a function (not just its 
 # source) as 'opp'.  If the hypothesis_string is a valid 
 # proof about that function, this function returns True; 
 # otherwise it returns False. For detailed assumptions on 
 # the proof system, see (Critch, 2019).
  ...
\end{minted}

Such proof-checking functionalities are available in programming languages such as HOL/ML \citep{nipkow2002isabelle,klein2009sel4} and Coq \citep{barras1997coq,chlipala2013certified} that are designed to enable formal verifiability, but in principle formal verifiers can be built for any programming language, including probabilistic programming languages.  HOL in particular has been used to formally verify proofs involving metamathematics and modal logic \citep{paulson2015mechanised,benzmuller2016inconsistency}, which might make it particularly suited to the kinds of proofs in this paper.  Some basic efficiency assumptions on the proof system are laid out in \ref{app:proofsystem}.

When trying to prove or disprove a given hypothesis about an opponent, how will our agents search through potential proofs?  For simplicity, we assume a trivial search method: the agents will search through finite strings in alphabetical order, and check each string to see if it is a proof.  Real-world formal verifiers use heuristics to more efficiently guide their search for proofs or disproofs of desired properties; we expect our results will apply in those settings as well.   For concreteness, and to eliminate any ambiguity as to what is meant in our examples and results, we exhibit the code for proof searching below: 

\begin{minted}{python}
# string_generator iterates lexicographically through all 
# strings up to a given length; see appendix:
from appendix import string_generator

# proof_search checks each generated proof string to see
# if it encodes a valid proof of a given hypothesis:
def proof_search(length_bound, opp_source, hypothesis):
  for proof in string_generator(length_bound):
    if proof_checker(opp_source, proof, hypothesis):
      return True
  # otherwise, if no valid proof of hypothesis is found:
  return False
\end{minted}

\subsection{Basic formal verifier agents}
A \Verb{proof_search} function like the one above can be used to design more interesting agents, who try for a while to formally verify properties about each other before deciding to cooperate or defect.

The following class of agents are called \Verb{CUPOD} for ``Cooperate Unless Proof of Defection'', defined by a function \Verb{CUPOD} that constructs an agent \Verb{CUPOD(k)} for each $k\ge 0$.  Specifically, \Verb{CUPOD(k)} is a particular agent that behaves as follows: \Verb{CUPOD(k)(opp_source)} searches for a proof, in $k$ characters or less, that the opponent \Verb{opp} is going to defect against \Verb{CUPOD(k)}.  If a proof of defection is found, \Verb{CUPOD(k)} defects; otherwise it ``gives the benefit of the doubt'' to its opponent and cooperates:

\newcommand{\cupod}{\Verb{CUPOD(k)}\xspace}
\begin{minted}{python}
# "Cooperate Unless Proof Of Defection" (CUPOD):
def CUPOD(k):
  def CUPOD_k(opp_source):
    if proof_search(k, opp_source, "opp(CUPOD_k.source) == D"):
      return D
    else:
      return C
  CUPOD_k.source = ... # the last 5 lines with k filled in
  return CUPOD_k
\end{minted}

\newcommand{\dupoc}{\Verb{DUPOC(k)}\xspace}
Next we have \Verb{DUPOC}, for ``Defect Unless Proof of Cooperation", the mirror image of \Verb{CUPOD}.  \Verb{DUPOC(k)} will defect unless it finds affirmative proof, in \Verb{k} characters or less, that its opponent is going to cooperate:

\begin{minted}{python}
# "Defect Unless Proof Of Cooperation" (DUPOC):
def DUPOC(k):
  def DUPOC_k(opp_source):
    if proof_search(k, opp_source, "opp(DUPOC_k.source) == C"):
      return C
    else:
      return D
  DUPOC_k.source = ... # the last 5 lines with k filled in
  return DUPOC_k
\end{minted}

\section{Results, Part I}\label{sec:results}
Our results are easiest to explain in the context of examples, and some of our open problems are easiest to explain in the context of results, so we deliver them all together.

\subsection{Basic properties of CUPOD and DUPOC}
Let us first examine some basic properties and examples of interactions with formal verifier agents, as a foundation for understanding the more difficult results.  

\begin{exm} \medium What is \Verb{outcome(CUPOD(10), DB)}? 
\answer{\Verb{CUPOD(10)(DB.source)} searches for a proof, using 10 characters or less, that \Verb{DB(CUPOD(10).source) == D}.  While it is true that \Verb{DB(CUPOD(10).source) == D}, the proof of this fact will take more than 10 characters of text to write down, so \Verb{CUPOD(10)} won't find the proof and will fall back on its default action. Hence, the outcome is \VCD.}
\end{exm}

\begin{exm} \medium What is \Verb{outcome(CUPOD(10^9), DB)}? Here \Verb{10^9} denotes $10^9$.
\answer{The fact that \Verb{DB(CUPOD(10^9).source) == D} is relatively easy to prove, via writing down a proof that \Verb{DB} always returns D; a hardworking undergrad could probably write out a fully rigorous proof by hand.  The length of the proof depends a bit on the specifics of the proof language we use to write down the proof, but for any reasonable proof language the proof shouldn't take more than a billion characters (around 200 thousand pages!).  So, \Verb{CUPOD(10^9)(DB.source)} will find the proof and return D, yielding the outcome \VDD.}
\end{exm}

These examples show the role of the length bound $k$.  Next, let us observe the following interesting property of \cupod: that it never defects on an opponent unless that opponent ``deserves it'' it in a certain sense:
\begin{proposition}\label{prop:cupod:doesntexploit} \cupod never exploits its opponent.  That is, for all opponents \Verb{opp} and all $k$,  \Verb{outcome(CUPOD(k),opp)} is never \Verb{(D,C)}.
\end{proposition}
\begin{proof} 
If \Verb{CUPOD(k)(opp.source)==D}, it must be because \Verb{CUPOD(k)} managed to prove that \Verb{opp(CUPOD(k).source)==D}.  Under the assumption that the \Verb{proof_check} function employed by \Verb{CUPOD(k)} is sound, it must be that \Verb{opp(CUPOD(k).source)==D}, so the outcome cannot be \VDC.
\end{proof}

Dually, we have the following:

\begin{proposition}\label{prop:dupoc:unexploitable} \dupoc is never exploited by its opponent.  That is, for all opponents \Verb{opp} and all $k$, \Verb{outcome(DUPOC(k),opp)} is never \Verb{(C,D)}.
\end{proposition}
\begin{proof}  The proof is the same as for Proposition \ref{prop:cupod:doesntexploit}, with \VC and \VD switched.
\end{proof}

\subsection{CUPOD vs CUPOD: What happens?}
Now let us examine our first example interaction between two open-source agents who both who both employ formal verifiers:

\vspace{1ex}
\begin{exm}\label{ex:cupod} \difficult What is \Verb{outcome(CUPOD(k), CUPOD(k))} when we set $k=10^{12}$?  
\end{exm}

We respectfully urge the reader to work through the earlier \examples before attending seriously to \Example \ref{ex:cupod}, to understand how the proof length bound can affect the answer.  These nuances represent a key feature of bounded rationality: when an agent thinks for a while about a hypothesis and reaches no conclusion about it, it still takes an action, if only a ``null'' or ``further delay'' action.

The next step is to notice what exactly \Verb{CUPOD(k)(CUPOD(k).source)} is seeking to prove and what it will do based on that proof.  Specifically, \Verb{CUPOD(k)(CUPOD(k).source)} is searching for a proof in $k$ characters or less that  \Verb{CUPOD(k)(CUPOD(k).source) == D}, and if it finds such a proof, it will return \VD.  Thus, we face a kind of circular dependency: the only way \Verb{CUPOD(k)(CUPOD(k).source)} will return \VD is if it \emph{first} finds a proof in fewer than $k$ characters that \Verb{CUPOD(k)(CUPOD(k).source)} will return \VD!

\textbf{Common approaches to resolving circularity.} How does this circular thinking in the \cupod vs \cupod interaction resolve?  Questions of this sort are key to any situation where two agents reasoning about each other interact, as each agent is analyzing the other agent analyzing itself analyzing the other agent [...].  Below are some common approaches to resolving how this circularity plays out for the \cupod{}s.

Approach 1 is to look for some kind of reduction to smaller values of $k$, as one might do when studying an iterated game.  But this game is not iterated; it's a one-shot interaction.  \Verb{CUPOD(k)(CUPOD(k).source)} will base its cooperation or defection on whether it proves \Verb{CUPOD(k)(CUPOD(k).source) == D}, not on the behavior of any \Verb{CUPOD(j)} for any smaller $j<k$.  This inability to reduce to smaller $k$ values might suggest the problem is in some sense intractable.  However, from the outside looking in, if the problem is so intractable that no short proof can answer it, then the result will have to be \VC,  because that's what happens when no shorter-than-k proof of \VD can be found.

Approach 2 is to wonder if both C and D are possible answers.  One can imagine the answer being \VC, so the algorithm finds a proof of D and therefore returns \VC, or the answer being D, so algorithm finds no proof of D and therefore returns \VC.  However, the programs involved here are deterministic, and have a bounded runtime! They must halt and return something.  So what do they return, \VC or \VD?  Again, if one cannot write down a shorter-than-$k$ proof of either answer, the answer must be \VC.

Approach 3 is to envision the circularity unrolling into a kind of stack overflow:  for defection to occur, \Verb{CUPOD(k)(CUPOD(k).source)} must prove that \Verb{CUPOD(k)(CUPOD(k).source)} proves that \Verb{CUPOD(k)(CUPOD(k).source)} proves that\ldots.  Such a proof would seem to be infinite in length, and thus would not fit within any finite length bound $k$.  This way of thinking gives an intuitive reason for why the answer should be \VC: the proof search for defection will find nothing, hit the length bound $k$, and then \Verb{CUPOD(k)(CUPOD(k).source)} will return \VC:

\emph{No-proof conjecture}: \Verb{CUPOD(k)(CUPOD(k).source) == C}, because there exists no finite-length-proof that \Verb{CUPOD(k)(CUPOD(k).source) == D}.

\subsection{Self-play results}

The first surprise of this paper is that the \emph{no-proof conjecture} directly above is false.  Because of a computationally bounded version of a theorem in logic called L\"{o}b's theorem \citep{critch2019parametric}, a fairly short proof that \Verb{CUPOD(k)(CUPOD(k).} \Verb{source) == D} in fact does exist, and as a result \Verb{CUPOD(k)(CUPOD(k).source)} finds the proof and defects because of it.   In other words, each of Approaches 1-3 for resolving the circular dependency in the previous section is incorrect:

\begin{theorem}\label{thm:cupod} For $k$ large, \Verb{outcome(CUPOD(k),CUPOD(k)) == (D,D)}.
\end{theorem}
\vspace{1ex}

\noindent The proof of this theorem, and several others throughout this paper, will make use of the lemma below, which depends on the following notation:

\begin{notation}[$S$, $\vdash$, $\Box$, and $\succ$]{\ }
\begin{itemize}
    \item $S$ stands for the formal proof system being used by the agents; see \ref{app:proofsystem} for more details about it.
    \item $\proves{} xyz...$ means ``the statement `xyz...' can be proved using the proof system $S$.
    \item $\bxof{k}{xyz...}$ means ``a natural number exists which, taken as a string, encodes a proof of $xyz...$, and that proof when written in the proof language of $S$ requires at most $k$ characters of text.''
    \item $f(k) \succ \Oo(\lg{k})$ means there is some positive constant $c>0$ and some threshold $\hat k$ such that for all $k>\hat k$, $f(k)>c\lg{k}$.
\end{itemize}
\end{notation}

\newcommand{\pbl}{\hyperref[lem:pblt]{PBLT}\xspace}
\begin{lemma}[\textbf{PBLT: Parametric Bounded L\"{o}b Theorem}]\label{lem:pblt}
Let $p[k]$ be a formula with a single unquantified variable $k$ in the proof language of the a proof system $S$, satisfying the conditions in \ref{app:proofsystem}.  Suppose that $k_1\in\NN$ and $f:\NN \to \NN$ is an increasing computable function satisfying $f(k) \succ \Oo(\lg{k})$, and $S$ can verify that the formula $p[k]$ is ``potentially self-fulfilling for large $k$'' in the sense that 
\[\proves{} \Forall{k>k_1}{\bxof{f(k)}{p[k]} \to p[k]}.\\
\]
Then there is some threshold $k_2 \in \NN$ such that
\[
\proves{} \Forall{k>k_2}{p[k]}.
\]
\end{lemma}

\begin{proof}[Proof of Lemma \ref{lem:pblt}]  This is a special case of \citet[Theorem 4.2]{critch2019parametric} where the proof expansion function $e(k)$ is in $\Oo(k)$.  Specifically, in this paper we have assumed (see \ref{app:proofsystem}) that a proof can be expanded and checked in time linear in the length of the proof, i.e., $e(k)=e^*\cdot k$ for some constant $e^*$.  This assumption makes the condition $f(k) \succ \Oo(\lg{k})$ sufficient to apply \citet[Theorem 4.2]{critch2019parametric}, as described by \citet[Section 4.2]{critch2019parametric}.
\end{proof}

\begin{proof}[Proof of Theorem \ref{thm:cupod}] This follows directly from \pbl, via the substitutions \\
$p[k]$=\Verb{(CUPOD(k)(CUPOD(k).source) == D)}, $f(k) = k$, and $k_1=0$.
\end{proof}

On the flip side of Theorem \ref{thm:cupod}, we also have the following symmetric result:

\begin{theorem}\label{thm:dupoc} For $k$ large, \Verb{outcome(DUPOC(k),DUPOC(k)) == (C,C)}.
\end{theorem}


\begin{proof} This follows directly from \pbl, via the substitutions \\
$p[k]=$\Verb{((DUPOC(k)(DUPOC(k).source) == C)}, $f(k) = k$, and $k_1=0$.
\end{proof}

This theorem may be equally surprising to Theorem \ref{thm:cupod}: most computer science graduate students guess that the answer is \Verb{D} when asked to predict the value of
\Verb{DUPOC(k)(DUPOC(k).source)}, even when given a few minutes to reflect and discuss with each other.\footnotemark{}  These guesses are almost always on the basis of one of the three misleading thinking approaches in the previous section, most commonly Approach 3, which expects the proof search to fail for stack-overflow-like reasons.

\footnotetext{Across several presentations to a total of around 100 computer science graduate students and faculty, when the audience is given around 5 minutes to discuss in small groups, around 95\% of attendees conjecture that the proof search would run out and the program would return \Verb{D}.}

\subsection{Intuition behind the proofs of Theorems \ref{thm:cupod} and \ref{thm:dupoc}}\label{sec:intuition}
Exactly what mechanism allows \dupoc to avoid these circularity and stack overflow problems?

A key feature of the proof of L\"{o}b's theorem, and \pbl which underlies Theorems \ref{thm:cupod} and \ref{thm:dupoc}, is the ability to a construct a statement that in some sense refers to itself.  This self-reference ability allows the proof to avoid the stack overflow problem one might otherwise expect.  In broad brushstrokes, the proof of \pbl follows a similar structure to the classical modal proof of L\"{o}b's theorem \citep{daryl2011modal}, which can be summarized in words for the case of Theorem \ref{thm:dupoc} as follows:
\begin{enumerate}
    \item We construct a sentence $\Psi$ that says ``If this sentence is verified, then the agents will cooperate,'' using a theorem in logic called the modal fixed point theorem.
    \item We show that if $\Psi$ can be verified by the agents, then mutual cooperation can also be verified by the agents, without using any facts about the agents' strategies other than their ability to find proofs of a certain length.
    \item We use the fact that verifying mutual cooperation causes the agents to cooperate, to show that $\Psi$ is true.
    \item Finally, we use the above proof of $\Psi$ to construct a formal verification of $\Psi$, which implies (by $\Psi$!) that cooperation occurs.
\end{enumerate}

This explanation, and the underlying mathematical proof, may be somewhat intuitively unsatisfying, because the meaning of the sentence $\Psi$ is somewhat abstract and difficult to relate directly to the agents.  Nonetheless, the result is true.  Moreover, we conjecture that a more illuminating proof may be possible:

\begin{definition} Let $S$ be the formal (proof) system used by \Verb{proof_checker}.  $\vdash X$ means ``the proof system $S$ can prove $X$'', and $\Box X$ means ``a natural number exists which encodes (via a G\"{o}del encoding) a proof of $X$ within the proof language and rules of $S$.''.
\end{definition} 

\begin{problem}\label{prob:lob}
L\"{o}b's theorem states that $\vdash(\Box C \to C)$ implies $\vdash(C)$.  We conjecture that L\"{o}b's Theorem can be proven without the use of the modal fixed point $\Psi \leftrightarrow (\Box\Psi \to C)$, by constructing an entire proof that refers to itself, according to the following intuitive template:
\begin{enumerate}
    \item  This proof is a proof of $C$.
    \item  Therefore $\Box C$.
    \item  By assumption, $\Box C \to C$.
    \item  Therefore, by (2) and (3), $C$.
\end{enumerate}
Is such a proof possible?  If not, why not?
\end{problem}

Such a proof would be suggestive of the following argument for one-shot cooperation between institutions, assuming each institution is DUPOC-like, in that it has already adopted internal policies, culture and personnel that will make it cooperate if it knows the other institution is going to cooperate:
\newpage

\begin{mdframed}
\vspace{1ex}
\noindent \underline{Cooperative affidavit for DUPOC-like institutions:}

\vspace{1ex}\noindent Institutions A and B have each recently undergone structural developments to prepare for cooperating with each other.  Moreover, representatives from each institution have thoroughly inspected the other institution's policies, culture, and personnel, and produced the attached inspection records with our findings, effectively rendering A and B ``open-source'' to one another.  These records show a readiness to cooperate from both institutions.  Moreover, the records are sufficient supporting evidence for the following argument:
\begin{enumerate}
    \item This signed document and the attached records constitute a self-evident (and self-fulfilling) prediction that Institutions A and B are going to cooperate.
    \item  Members of Institutions A and B can all read and understand this document and attached records, and can therefore tell that the other institution is going to cooperate.
    \item  Institution A's internal policies and culture are such that, upon concluding that Institution B is going to cooperate, Institution A will cooperate.  The same is true of Institution B's policies and culture with regards to Institution A.
    \item  Therefore, by (2) and (3), the Institutions A and B are going to cooperate.
\end{enumerate}
\end{mdframed}

Points (1)-(4) of this ``cooperative affidavit'' are manifestly a kind of circular argument.  However, the lesson of \pbl is that certain kinds of circular arguments about bounded reasoners can be logically valid.  In this case, most of the ``work'' toward achieving cooperation is carried out in each institution's adoption of DUPOC-like internal culture and policies (``we cooperate once we know they're going to cooperate''), and in the furnishing of mutual inspection records adequate to create common knowledge of that fact.  Once that much is done, the validation of the circular argument for cooperation becomes both logically feasible and sufficient to trigger cooperation.

It is also interesting to note that the above ``cooperative affidavit'' is not a contract; it is merely a signed statement of fact made valid by each company's DUPOC-like internal policies and culture.  As such, the document would add no mechanism of enforcement to ensure cooperation; it would simply trigger each institution's internal policies and culture to take effect in a cooperative manner, assuming those policies and culture are sufficiently DUPOC-like.  In Section \ref{sec:prophecies} we will discuss this prospect further in the context of existing literature spanning social psychology, economics, and education.

Finally, we remark that a similar phenomenon could result in CUPOD-like institutions turning out to defect against each other.  In other words, when two institutions are cooperating primarily on the basis of failing to understand each other, but stand ready to defect if they gain enough information about the counterparty to confidently predict defection, then a self-fulfilling prophecy of defection can arise as soon as the institutions believe they have become sufficiently ``open-source''.

Thus, CUPOD-like institutions could turn out to defect, and DUPOC-like institutions can turn out to cooperate, as a result of self-fulfilling prophecies analogous to Theorems \ref{thm:cupod} and \ref{thm:dupoc}.  This suggests that  institutions who are going to (intentionally or inadvertently) reveal their inner workings should consider first shifting from CUPOD-like policies and culture to DUPOC-like policies and culture.

\subsection{Further properties of \dupoc}
Another counterintuitive observation is that \cupod seems ``nicer'' than \dupoc, since by Propositions \ref{prop:cupod:doesntexploit} and \ref{prop:dupoc:unexploitable}, \cupod never exploits its opponent, while \dupoc is more ``defensive'' in that it never allows itself to be exploited.  This observation has given many readers the intuition that \cupod would have an easier time achieving mutual cooperation than \dupoc, but by Theorems \ref{thm:cupod} and \ref{thm:dupoc}, the opposite turns out to be true.  In fact, the agents \Verb{DUPOC(k)} for large values of $k$ --- henceforth ``DUPOCs'' --- have some very interesting and desirable game-theoretic properties:

\newcommand{\itemname}[1]{\item \textbf{(#1)} }
\begin{enumerate}
    \itemname{unexploitability} A DUPOC never cooperates with an opponent who is going to defect, i.e., it never receives the outcome \VCD in a game; only \VDD, \VDC, or \VCC.

    \itemname{broad cooperativity} \label{item:cc} DUPOCs achieve \VCC not just with other DUPOCs, but with a very broad class of agents, as long as those agents follow a certain basic principle of ``G-fairness'' described in \citet{critch2019parametric}.  Roughly speaking, if an opponent always cooperates with agents who exhibit legible (easy-to-verify) cooperation, then DUPOCs cooperate with that opponent.

    \itemname{rewarding legible cooperation} In order to achieve \VCC with a DUPOC, the opponent must not only cooperate, but \emph{legibly cooperate}, i.e., cooperate in a way that the DUPOC is able to verify with its bounded proof search.  So, if there are opponents around that are made of ``spaghetti code'' that the DUPOCs are unable to efficiently analyze, those opponents won't be rewarded.
    
    \itemname{rewarding \emph{rewarding} legible cooperation} Moreover, by item \ref{item:cc} above, the presence of DUPOCs in a population also rewards \emph{rewarding} legible cooperation, by cooperating with all agents who do it.

\end{enumerate}

Because of these points, one can imagine DUPOCs having a powerful effect on any population that contains them:

\vspace{1ex}
\begin{problem}\label{prob:modalpopulation} Determine conditions under which a population containing DUPOCs will evolve into a population where all agents reward legible cooperation, i.e., are \emph{G-fair} in the sense defined by \citet{critch2019parametric}.
\end{problem}

Note that a DUPOC does not cooperate with every agent who has some tendency to cooperate; in fact we have the following open problem:

\vspace{1ex}
\begin{problem}
\label{prob:dupoccupod} For large values of $k$, we conjecture that \\
\Verb{outcome(DUPOC(k),CUPOD(k))==(D,C)}.  Is this the case?
\end{problem}
\begin{challenging} 
A natural approach would be to argue by symmetry that the result must be \VCD or \VDC, and rule out \VCD on the grounds that \dupoc is unexploitable.  However, the proposed symmetry argument isn't quite valid: \dupoc is not a perfect mirror image of \cupod, because the letters \VC and \VD are not switched inside the \Verb{proof_checker} subroutine.  It seems we need some way to reason about one agent's proof search running out, while the other agent is unable to prove that the first agent's proof search runs out.
\end{challenging}

%
%
%

\section{Generalizability of results}
\subsection{Nondeterministic interactions}
A common question to ask about these results is whether they depend on the `rigidity' of logic in some way.  The answer is, roughly speaking, no.  To see this, let us examine a nondeterministic version of DUPOC, using a "\Verb{Prob}" symbol to write proofs about the probability of random events.  This agent tries to show that the opponent cooperates with probability at least $q$, and if successful, cooperates with probability $q$:

\newcommand{\pdupoc}{\Verb{PDUPOC(k)}\xspace}
\begin{minted}{python}
# Probabilistic DUPOC (PDUPOC):
def PDUPOC(K,Q):
  def PDUPOC_k(opp_source):
    k = K; q = Q
    if proof_search(k, opp_source, 
    "Prob('opp(DUPOC_k.source) == C' >= q")):
      if random.uniform(0,1) <= q:
        return C
    else:
      return D
  PDUPOC_k.source = ... # the last 8 lines with K,Q filled in
  return DUPOC_k
\end{minted}

\begin{theorem}\label{thm:pdupoc}
For large $k$ and $q\geq 0.5$, with probability at least $2q-1$, \\
\Verb{outcome(PDUPOC(k,q),PDUPOC(k,q))==(C,C)}.  
In particular, for $q \approx 1$, the probability of mutual cooperation is also $\approx 1$. 
\end{theorem}

\begin{proof} This follows directly from \pbl, using the substitutions \\
$p[k]=$\Verb{(Prob('PDUPOC(k,q)(PDUPOC(k,q).source)==C') >= q)}, $f(k) = k$, and $k_1=0$.  We use the fact that two events with probability $q$ must have a conjunction with probability at least $2q-1$.  If the sources of randomness are known to be independent, the bound $2q-1$ can be increased to $q^2$.
\end{proof}

\subsection{More efficient \Verb{proof_search} methods}\label{sec:efficiency}

For simplicity of analysis, we assumed that our \Verb{proof_search} method searches alphabetically through arbitrary strings until it finds a proof, which is terribly inefficient, taking exponential time in the length of the proof.  This problem has been noted as potentially prohibitive~\citep{oesterheld2019robust}.

However, much more efficient strategies are possible.  As a trivial example, if \Verb{string_generator} is reconfigured to output, as its first string, the proof of Theorem \ref{thm:dupoc}, then both agents will terminate their searches in a tiny fraction of a second and cooperate immediately!  
Indeed, the proof of \pbl makes no use of the order in which potential proofs are examined.  In reality, the details of the proofs employing \pbl will need to vary based on the particularities of the implementations of each agent, so heuristic searches such as those used to complete proof goals in hybrid proof-and-programming languages such as HOL \citep{nipkow2002isabelle,klein2009sel4} and Coq \citep{barras1997coq,chlipala2013certified} will likely be crucial to automating these strategies in real-world systems.  And of course, there is a trade-off between flexibility on implementation standards and speed of verification.  This poses an interesting research project:

\begin{problem}\label{prob:holdupoc}
Implement \Verb{DUPOC} using heuristic proof search in HOL/ML or Coq.  Can \Verb{outcome(DUPOC(k),DUPOC(k))} run and halt with mutual cooperation on a present-day retail computer?  We conjecture the answer is yes.  If so, how much can the implementations of the two agents be allowed to vary while cooperative halting is preserved?
\end{problem}

\section{Results, Part II: Strategies that verify conditionals}

\subsection{Rewarding conditional cooperation: CIMCIC}
Notice how a DUPOC's strategy is to prove that the opponent is going to cooperate \emph{unconditionally}.  In other words, a DUPOC doesn't check ``If I cooperate, then the opponent cooperates''; rather, it checks the stronger statement ``The opponent cooperates''.  The following is the ``conditional cooperation'' version of DUPOC:

\begin{minted}{python}
# "Cooperate If My Cooperation Implies Cooperation  from
# the opponent" (CIMCIC):
def CIMCIC(k):
  def CIMCIC_k(opp_source):
    if proof_search(k, opp_source,
    "(CIMCIC_k(opp.source)==C) => (opp(CIMCIC_k.source)==C)"):
      return C
    else:
      return D
  CIMCIC_k.source = ... # the last 6 lines with k filled in
  return CIMCIC_k
\end{minted}

\begin{proposition} CIMCIC is unexploitable, i.e., \Verb{outcome(CIMCIC(k),opp)} is never \Verb{(C,D)} for any any opponent. 
\end{proposition}

\begin{proof}
This proof works the same as that of Proposition \ref{prop:cupod:doesntexploit}.  Specifically, 
if \Verb{CIMCIC(k)(opp.source)==C}, it must be because \Verb{CIMCIC(k)} managed to prove \Verb{CIMCIC(k)}'s cooperative criterion, namely \\
\indent \Verb{(opp(CIMCIC(k).source)==C) => (opp(CIMCIC_k.source)==C)}.\\
Under the assumption that the \Verb{proof_check} function employed by \Verb{CIMCIC(k)} is sound, this implication must hold, implying \Verb{opp(CIMCIC(k).source)==C}, so the outcome cannot be \VCD.
\end{proof}

\vspace{1ex}
\begin{exm} Evaluate the following (answered below):
\begin{enumerate}[label=\alph*)]
    \item \medium \Verb{Outcome(CIMCIC(k),CIMCIC(k))}
    \item \difficult \Verb{Outcome(DUPOC(k),CIMCIC(k))}
\end{enumerate}
\end{exm}

Since CIMCIC's cooperative criterion is \emph{conditional}, one can imagine two CIMCICs being in a kind of stand-off, with each one thinking ``I'll cooperate if I know it will imply you cooperate'', but never achieving cooperation because there are no provisions in its code to break the cycle of ``I'll do it if you will'' reasoning.  However, the following theorem shows that, by \pbl, no such provision is necessary:  

\begin{theorem}\label{thm:cimcic} For large $k$, 
\begin{enumerate}[label=\alph*)]
    \item \Verb{outcome(CIMCIC(k),CIMCIC(k)) == (C,C)}
    \item \Verb{outcome(DUPOC(k),CIMCIC(k)) == (C,C)}
\end{enumerate}
\end{theorem}
\begin{proof}[Proof of (a)] A short proof of the outcome \VCC will lead, in a few additional lines comprising some number of characters $c$, to a proof of the material implication
\vspace{-1.1ex}
\begin{minted}{python}
"(CIMCIC_k(CIMCIC_k.source)==C)=>(CIMCIC_k(CIMCIC_k.source)==C)"
\end{minted}
\vspace{-1.1ex}
This in turn will cause the agents to cooperate, bringing about the outcome \VCC.  Thus, we have a ``self-fulfilling prophecy'' situation of the kind where \pbl can be applied.  Specifically, if we let $f(k)=\floor{k/2}$, $k_1=2c$, and $p[k]$ be the statement
\vspace{-1.1ex}
\begin{minted}{python}
"outcome(CIMCIC(k),CIMCIC(k))==(C,C)"
\end{minted}
\vspace{-1.1ex}

\noindent then the conditions of \pbl are satisfied.  Therefore, for some constant $k_2$, we have $\vdash \Forall{k>k_2}{p[k]}$.
\end{proof}

\begin{proof}[Proof of (b)] Again, a short proof of the outcome \VCC will lead, in a few additional lines comprising some number of characters $c$, to a proof of CIMCICs' cooperation condition, namely
\vspace{-1.1ex}
\begin{minted}{python}
"(CIMCIC_k(DUPOC_k.source)==C)=>(DUPOC_k(CIMCIC_k.source)==C)"
\end{minted}
\vspace{-1.1ex}
as well as DUPOC's cooperation condition, 
\vspace{-1.1ex}
\begin{minted}{python}
"CIMCIC(k)(DUPOC(k).source)==C"
\end{minted}
\vspace{-1.1ex}
Thus, letting $f(k)=\floor{k/2}$, $k_1=2c$, and $p[k]$ be the statement
\vspace{-1.1ex}
\begin{minted}{python}
"outcome(DUPOC(k),CIMCIC(k))==(C,C)"
\end{minted}
\vspace{-1.1ex}
satisfies the conditions of \pbl.  Therefore, for some constant $k_2$, we have $\vdash \Forall{k>k_2}{p[k]}$.
\end{proof}

\begin{problem}\label{prob:cupodcimcic}
What is \Verb{outcome(CUPOD(k),CIMCIC(k))}?
\end{problem}
\newcommand{\cimcic}{\Verb{CIMCIC(k)}\xspace}

\begin{challenging} Write $CC, CD, DC,$ and $DD$ as shorthand for \Verb{(C,C)}, \Verb{(C,D)}, \Verb{(D,C)}, and \Verb{(D,D)} respectively, and write $C*$ for ``$CC\textrm{ or }CD$'' and $*C$ for ``$CC\textrm{ or }DC$''.  We know that the outcome $DC$ is impossible, because \cupod never exploits its opponent.  However, the other possibilities are difficult to prove or rule out. If the outcome is $CC$, it means \cimcic manages to prove mutual cooperation, which is consistent with \cupod failing to prove defection.  If the outcome is $DD$, then it means \cupod manages to prove that \cimcic defects, but somehow \cimcic fails to find a short proof \emph{of the fact that \cupod proves that \cimcic defects}. Otherwise \cimcic could prove the material conditional $*C \to C*$ with just a few more lines, which would cause it to cooperate, a contradiction.  If the outcome is $CD$, it means both agents' proof searches run out.  It's not immediately clear which of these outcomes will actually obtain.
\end{challenging}

\subsection{Preempting exploitation: DIMCID}
Another interesting strategy to consider is to preempt exploitation by defecting if the opponent would exploit you:

\begin{minted}{python}
# "Defect if My Cooperation Implies Defection from 
# the opponent (DIMCID):
def DIMCID(k):
  def DIMCID_k(opp_source):
    if proof_search(k, opp_source, 
    "(DIMCID_k(opp.source)==C) => (opp(DIMCID_k.source)==D)"):
      return D
    else:
      return C
  DIMCID_k.source = ... # the last 6 lines with k filled in
  return DIMCID_k
\end{minted}

\begin{exm} \easy Evaluate the following \answer{\VCC and \VDD, respectively.}: 
\begin{enumerate}[label=\alph*)]
    \item \Verb{outcome(DIMCID(k), CB)}
    \item \Verb{outcome(DIMCID(k), DB)}
\end{enumerate}
\end{exm}
\vspace{1ex}
\begin{exm} \medium Evaluate the following (answered below):
\begin{enumerate}[label=\alph*)]
\item \Verb{outcome(DIMCID(k), DIMCID(k))}
\item \Verb{outcome(CUPOD(k), DIMCID(k))}
\end{enumerate}
\end{exm}

\vspace{1ex}

\begin{theorem}\label{thm:dimcid} For large $k$,

\begin{itemize} 
    \item \Verb{outcome(DIMCID(k), DIMCID(k)) == (D,D)}
    \item \Verb{outcome(CUPOD(k), DIMCID(k)) == (D,D)}.
\end{itemize}
\end{theorem}

\begin{proof}[Proof of (a)] A short proof of the outcome \VDD will lead, in a few additional lines comprising some number of characters $c$, to a proof of the material implication
\vspace{-1.1ex}
\begin{minted}{python}
"(DIMCID_k(DIMCID_k.source)==C)=>(DIMCID_k(DIMCID_k.source)==D)"
\end{minted}
\vspace{-1.1ex}
which in turn will cause the agents to defect, bringing about the outcome \VDD.  Thus, we have a ``self-fulfilling prophecy'' situation of the kind where \pbl can be applied.  Specifically, if we let $f(k)=\floor{k/2}$, $k_1=2c$, and $p[k]$ be the statement
\vspace{-1.1ex}
\begin{minted}{python}
"outcome(DIMCID(k),DIMCID(k))==(D,D)"
\end{minted}
\vspace{-1.1ex}

\noindent then the conditions of \pbl are satisfied.  Therefore, for some constant $k_2$, we have $\vdash \forall k>k_2, p[k]$.
\end{proof}

\begin{proof}[Proof of (b)] Again, a short proof of the outcome \VDD will lead, in a few additional lines comprising some number of characters $c$, to a proof of DIMCID's defection condition, namely
\vspace{-1.1ex}
\begin{minted}{python}
"(DIMCID_k(CUPOD_k.source)==C)=>(CUPOD_k(DIMCID_k.source)==D)"
\end{minted}
\vspace{-1.1ex}
as well as CUPOD's defection condition, 
\vspace{-1.1ex}
\begin{minted}{python}
"DIMCID(k)(CUPOD(k).source)==D"
\end{minted}
\vspace{-1.1ex}
Thus, letting $f(k)=\floor{k/2}$, $k_1=2c$, and $p[k]$ be the statement
\vspace{-1.1ex}
\begin{minted}{python}
"outcome(CUPOD(k),DIMCID(k))==(D,D)"
\end{minted}
\vspace{-1.1ex}
satisfies the conditions of \pbl.  Therefore, for some constant $k_2$, we have $\vdash \Forall{k>k_2}{p[k]}$.
\end{proof}

\begin{problem}\label{prob:dupocdimcid} What is \Verb{outcome(DUPOC(k),DIMCID(k))}?
\end{problem}
\begin{challenging} This problem is very similar to Open Problem \ref{prob:cupodcimcic}: one of the agent's proof searches might fail while the other succeeds.  If so, which way does it go?  If both proof searches fail, why is that?
\end{challenging}

\section{Discussion}

\subsection{Applicability to non-logical reasoning processes}

Must the agents in these theorems form their beliefs using logical proofs, or would the same results hold true if they employed other belief-forming procedures, such as machine learning?  What abstract properties of these proofs and proof systems are key to applying \pbl in Theorems \ref{thm:cupod},  \ref{thm:dupoc},  \ref{thm:pdupoc}, \ref{thm:cimcic}, and \ref{thm:dimcid}?  

The following features of the agents' reasoning capabilities are key to the proof of \pbl, and hence to the results of this paper:
\begin{enumerate}
  \item The proof language for representing the agents' beliefs needs to be expressive enough to represent numbers and computable functions and to introduce and expand abbreviations.
  \item There must be a process the agents can follow for deriving beliefs from other beliefs (writing a proof is such a process).
  \item  The proof language must also be able to represent belief derivations (e.g., proofs) as objects in some way, and refer to the cost (e.g., in time or space) of producing various belief derivations.
\end{enumerate}

Agents lacking some of these reasoning capabilities should not be expected to behave in accordance with the theorems of this paper.  However, these capabilities are all features of \emph{general intelligence}, in the sense that a reasoning process that lacks these capabilities can be made more deductively powerful by adding them.  Thus, there is some reason to expect that agents designed to exhibit highly general reasoning capabilities may be amenable to these results.  The task of formalizing highly general definitions of ``beliefs'' and ``belief derivations'' could itself make for interesting follow-up work.

\subsection{Further open problems}
\newcommand{\problemarea}[1]{\textbf{In #1:}}

As mentioned earlier in Section \ref{sec:rw}, \citet{barasz2014robust} exhibit a finite-time algorithm written in Haskell for settling interactions between agents that write proofs about each other, which unfortunately assumes the agents themselves have unbounded (infinite) computation with which to conduct their proof searches.  
The Haskell algorithm---based on Kripke semantics for \godel--\lob provability logic \citep[Chapter 4]{boolos1995logic}---is sound for settling these unbounded cases using abstract mathematics and is efficient enough to run tournaments between dozens of their abstractly specified agents, which they call \emph{modal agents}.  The results of \citet{critch2019parametric} go some of the way toward generalizing their results for bounded versions of very simple modal agent interactions, but not for all interactions, such as those in Open Problems \ref{prob:dupoccupod}, \ref{prob:cupodcimcic}, and \ref{prob:dupocdimcid}.

Thus, it would be quite interesting to work out the details of how and when theorems in the logic used by \citet{barasz2014robust} (\godel--\lob provability logic) have analogues with bounded proof lengths, and if they could really be applied to make resource-bounded (and hence physically realizable) analogues of all the modal agents of \citet{barasz2014robust} and \citet{lavictoire2014program}.  At this level of generality, it would be useful to also carry through those details to ascertain concrete proof lengths needed for any given game outcome, so their Haskell algorithm could be modified to return the value of $k$ needed for each proof search to find the requisite proofs:

\begin{problem}\label{prob:bgl} Generalize \pbl to a bounded analogue of \godel--\lob that systematically tracks the proof lengths necessary or sufficient for each instance of the provability operator $\Box$.
\end{problem}

\begin{problem}\label{prob:haskellupgrade}
Apply the result of Open Problem \ref{prob:bgl} to improve the Haskell algorithm at \url{github.com/klao/provability}, so that it returns the proof length needed for each proof search conducted by each agent in a given open-source interaction.
\end{problem}

\citet{lavictoire2014program} exhibit a computationally unbounded agent called PrudentBot, which is similar to DUPOC except that it uses employs an additional proof search that allows it to defect against CooperateBot.  Such agents are particularly interesting at a population scale because they have the potential to drive CooperateBots out of existence, which in turn would make it more difficult for DefectBot to survive.  Hence we ask:

\vspace{1ex} 
\begin{problem}\label{prob:prudentbot}
Does a computationally bounded version of PrudentBot \citep{lavictoire2014program} exist?
\end{problem}

If so, questions regarding population dynamics among CooperateBots, DefectBots, DUPOCs, and PrudentBots would be interesting to examine, especially if the game payoffs take into account the additional cost of proof-searching incurred by PrudentBot.

\newcommand{\VE}{\Verb{E}\xspace}
Finally, to begin generalizing this theory to games with more than two actions, consider next an extended Prisoners' Dilemma with a third option, \VE for ``encroachment'', that is even more tempting and harmful than defection:

\begin{game}{3}{3}
      & $C$     & $D$   & $E$       \\
$C$   & $2,2$   & $0,3$ & $-2,4$     \\
$D$   & $3,0$   & $1,1$ & $-1,2$ \\
$E$   & $4,-2$  & $2,-1$ & $0,0$
\end{game}

\medskip

\newcommand{\eupod}{\Verb{EUPOD(k)}\xspace}
\noindent In such a setting, it is natural for an agent to first target mutual cooperation, and failing that, target mutual defection, and failing that, encroach.   If the agent uses proof searches to ``target'' outcomes, this means it would first cooperate if a proof of \VCC can be found, and failing that, defect if a proof of \VDD can be found, and failing that, encroach.  Let's call this vaguely described agent concept ``CDEBot.''  An interesting opponent for CDEBot is \eupod, for ``Encroach Unless Proof of Defection,'' which is just like \dupoc with with \VD and \VE respectively in place of \VC and \VD.

\begin{problem}\label{prob:cdebot}
Implement a version of CDEBot using bounded proof searches, ensuring that it achieves \VCC with \dupoc, and \VDD with \eupod.
\end{problem}

\begin{challenging} CDEBot's proof search for \VDD needs to be somehow predicated on the failure of the proof search for \VCC by one or both of the agents, and it's not immediately clear how \pbl can be applied in that setting.
\end{challenging}

\subsection{Applicability to humans and human institutions: modelling self-fulfilling prophecies}\label{sec:prophecies}

How and when might our framework be applicable to interactions between individual humans or human institutions who reason about themselves and each other?  

One important and immediate observation is that the open-source condition provides a model under which certain \emph{self-fulfilling prophecies} will turn out to be true, by virtue of the single ``speech act'' of each agent's source code being revealed.  Specifically, our main theorems (\ref{thm:cupod}, \ref{thm:dupoc}, 
\ref{thm:pdupoc}, 
\ref{thm:cimcic}, and 
\ref{thm:dimcid})  each involve a kind of self-fulfilling prophecy: upon the agents verifying that a certain outcome is going to happen, they choose to make it happen, on the basis of the verification.  For \dupoc{}s, the prophecy is ``made'' once each agent has successfully written down a proof that they will cooperate, and it ``comes true'' in their reaction to that proof (deciding to cooperate).

Self-fulfilling prophecies are known to be an important phenomenon in the social sciences, as Milton \citet{friedman1977nobel} remarked in his Nobel Lecture:

\begin{quote}
``Do not the social sciences, in which scholars are analyzing the behavior of themselves and their fellow men, who are in turn observing and reacting to what the scholars say, require fundamentally different methods of investigation than the physical and biological sciences? Should they not be judged by different criteria?  I have never myself accepted this view. [...]  In both, there is no way to have a self-contained closed system or to avoid interaction between the observer and the observed. The G\"{o}del theorem in mathematics, the Heisenberg uncertainty principle in physics, the self-fulfilling or self-defeating prophecy in the social sciences all exemplify these limitations.''
\end{quote}

\noindent Indeed, as Friedman alludes, \godel's incompleteness theorems are special cases of L\"{o}b's theorem (itself a special case of \pbl) where the self-fulfilling prophecy $p$ is a contradiction.  L\"{o}b's theorem essentially says ``If a particular statement $p$ \emph{would be} self-fulfilling (i.e., $\Box(p) \to p$) then it \emph{does} self-fulfill ($\Box(p)$ follows, and then $p$ follows).''  The logician Raymond \citet{smullyan1986logicians} also noted, in ``Logicians who reason about themselves,'' that L\"{o}b's theorem ``reflects itself in a variety of beliefs which of their own nature are necessarily self-fulfilling.''  He tells a story of a cure that works only if the recipient believes in it, and argues by L\"{o}b's theorem that the cure is therefore bound to work.  \pbl shows that this kind of argument can really be applied in the minds of bounded reasoners.

In which areas of social science do self-fulfilling prophecies play an important role?  Each of the following disciplines have included analyses of the conditions and mechanisms by which self-fulfilling prophecies can arise and affect the world, and might therefore present interesting domains in which to apply open-source agent models and \pbl more specifically:
\begin{itemize}
    \item In social psychology, \citet{merton1948self,kelley1970social,word1974nonverbal,snyder1977social,darley1980expectancy}, and at least a dozen other authors have examined conditions and mechanisms whereby believing in things can cause them to become true, such as stereotypes or interracial conflicts.
    \item In education research, \citet{rist1970student,wilkins1976concept,brophy1983research,wineburg1987self,jussim1996social} and \citet{jussim2005teacher} have examined how and when teachers' positive or negative expectations of students can cause those expectations to be fulfilled by students.
    \item In the study of management and leadership, \citet{ferraro2005economics} have examined how theories of economics and social science can cause their own validity, and  \citet{eden1992leadership,eden1984self} has examined how leaders' expectations can cause organization members to conform to those expectations.
    \item In political and international relations theory, \citet{zulaika2003self,bottici2006rethinking,verhoeven2009self,houghton2009role,frisell2009theory,ahler2014self} and \citet{swiketek2017yemen} have examined how peace and conflict can both arise as self-fulfilling prophecies.
\end{itemize}

In economics specifically, there has been considerable work examining the potential causes and effects of self-fulfilling prophecies.  Typically these studies employ dynamical systems models of repeated interactions between players, or equilibrium assumptions that do not specify how exactly the players manage to know each other's strategy.  By contrast, the open-source condition makes it clear and concrete how each player's strategy could be known to the others, and the theorems of this paper show that open-source agents do not require repeated interaction for self-fulfilling prophecies to manifest.   Being open-source can also be seen as a limiting case of an agent being nontrivially but imperfectly understood by its opponents.  Thus, in real-world institutions or software systems that are only imperfectly visible to each other, one-shot, self-fulfilling prophecies might still arise from merely probabilistic mutual knowledge, as in Theorem \ref{thm:pdupoc}.

When do economic self-fulfilling prophecies occur in the real world? Numerous authors have argued that currency crises can arise from self-fulfilling prophecies \citep{obstfeld1996models,morris1998unique,metz2002private,hellwig2006self,guimaraes2007risk}. Relatedly, 
\citet{obstfeld1986rational} showed conditions under which  balance-of-payments crises may be ``purely self-fulfilling events rather than the inevitable result of unsustainable macroeconomic policies.''
\citet{azariadis1981self} exhibits a model of intergenerational production and consumption in which self-fulfilling prophecies constitute between one third and one half of the possible pricing equilibria.
Others have argued that self-fulfilling prophecies can give rise to 
debt crises \citep{cole2000self},
credit market freezes \citep{bebchuk2011self}, and liquidity dry-ups \citep{malherbe2014self}. On the other hand, \citet{krugman1996currency} argues that ``the actual currency experience of the 1990s does not make as strong a case for  self-fulfilling crises as has been argued by some researchers. In general, it will be very difficult to distinguish between crises that need not have happened and those that were made inevitable by concerns about future viability that seemed reasonable at the time.''  

In our view, these observations and debates are an invitation to develop a new and more rigorous formalism for the study of bounded agents who reason about themselves and each other.  As can be seen from our main theorems, such a formalism can be used to construct and model self-fulfilling prophecies using \pbl, even in single-shot games with no repetition of interactions between players.  Going further, understanding self-fulling prophecies should not be the extent of our ambition with open-source game theory.  For instance, we should also address self-defeating prophecies, as alluded to by \citet{friedman1977nobel}, or prophecies that plausibly-but-don't-quite self-fulfill, as noted by \citet{krugman1996currency}.

\subsection{Comparison to Nash Equilibria}
In this section we argue that, in real-world single-shot games, the conditions under which a genuine Nash (or correlated) equilibrium could arise are often conditions under which an open-source game (or commitment game) could be played instead, enabling more desirable equilibria.  In short, the reason is that the Nash equilibrium concept already assumes a certain degree of mutual transparency, and often that transparency could instead be used to reveal programs/plans/commitments.

For a more detailed argument, let us first reflect on the following important observation due to Moulin regarding Nash equilibria:
\begin{quote}
``Nash equilibria are self-fulfilling prophecies. If every player guesses what strategies are chosen by the others, this guess is consistent with selfish maximization of utilities if and only if all players bet on the same Nash equilibrium.  Here, we need an invisible mediator or some theory that pronounces x the rational strategic choice for Player i.''  
\citep[Ch. 5]{moulin1986game} 
\end{quote}

\noindent How does this ``pronouncement'' play out in the minds of the players?  Suppose in a game with two real-world players, before any discussion about what equilibrium they should play, Player $i$'s subjective probability distribution over Player $j$'s action is denoted $a^i_j(0)$.  
\newcommand{\br}{\textrm{BR}}
Let $BR_i$ denote player $i$'s best response function, so $b_i(0):=\br_i(a^i_{-i}(0))$ is the action distribution Player $i$ will use if no discussion happens and no prescribed equilibrium is ``pronounced.''  Then, suppose a discussion or mediation occurs between the players in which a particular  Nash equilibrium $x=(x_1,x_2)$ is prescribed to be played instead of the defaults $(b_1(0),b_2(0))$ that they would have played otherwise.  How can each player $i$ become convinced that the other is really going to play $x_{-i}$, before their action is taken?    There are (at least) two perspectives commonly taken on this question.  

The first perspective is simply to say that each player has ``no reason to deviate'' from the prescription:  ``I am not forced to follow our agreement, but as long as you guys are faithful to it, I have no incentive to betray'' \citep[Introduction]{moulin1986game}.  But this does not answer the question: what reason do the players have to believe the prescription will be carried out to begin with?  
Player 1 can think, ``If Player 2 is thinking that I (Player 1) will play $x_1$, then Player 2 will play $x_2$''... but how would Player 1 know that Player 2 knows that Player 1 is going to do that?  
Indeed, Player 1 might reason, ``My posterior says Player $2$ \emph{might} play the prescribed distribution $x_2$, but they might instead do some other thing, such as a `safe' action, or sampling from some prior distribution over actions.  Reflecting on this, my posterior over Player 2's action has updated to a new value, $a^i_j(1)$, and my best response is $b_1(1):=\br_1(a^1_{2}(1))$.''  Player 2 can update similarly, yielding sequences $a^i_j(n)$ and $b_i(n)$ as the number of rounds $n$ of discussion and reflection grows. This process is different from a simple best-response iteration, because there is also a discussion component giving rise to each $a^i_j(n)$.   Is there any reason to think it will converge on a fixed point?   And if so, does it converge on the ``pronounced'' equilibrium $(x_1,x_2)$, or some other point $(x'_1,x'_2)$?

These questions inspire a second perspective for applying the Nash equilibrium concept, which is to simply \emph{define} $(x_1,x_2)$ to be whatever the players converge to deciding and believing about each other through this discussion process, \emph{assuming it converges}.  In this view, if a mediator or compelling bargaining theory is available, perhaps it can help to ``steer'' the conversation toward a particular desirable equilibrium, rather than merely stating it once and assuming conformity.  

In both perspectives, the players \emph{end up} in a state of mutual understanding, where each is satisfied (at equilibrium) with her model of the other.  Indeed:
\begin{quote}
``Some kind of communication among the players is necessary to endow them with mutually consistent beliefs, and/or allow mutual observation of past outcomes.''
\citep[Ch. 5]{moulin1986game}
\end{quote}

\noindent In other words, by the end of the discussion, the players no longer have full privacy: they understand each other to some extent.  Why not use that understanding to make and reveal commitments to perform better than the Nash equilibrium in games like the Prisoner's Dilemma?  
That is to say, why don't the players simply make internal structural changes in the form of unexploitable cooperative commitments ({\`a} la Theorem \ref{thm:dupoc}), and then reveal those changes through the same process that would have allowed them to understand each other in the Nash case?  

If the players are either institutions or software systems capable of self-repair, such self-imposed changes are both possible and potentially legible.  If we suppose they make themselves fully transparent, a question still arises to whether their discussion of plans will converge, but as we've seen that question can sometimes be resolved using \pbl.  Specifically, when there appears to be an infinite regress of social metacognition, sometimes that regress can be collapsed into a single self-reflective observation that the equilibrium is going to be observed, along the lines of Theorem \ref{thm:dupoc}.

In summary, the kinds of mutual transparency adequate to allow Nash equilibria in one-shot games might often allow the players to reveal enough of their internal structure or ``source codes'' to enable open-source equilibria or commitment game equilibria to arise instead.

\subsection{Conclusion}
Our most important conclusion is that fundamental questions regarding the interaction of open-source agents should not be left unaddressed.  As the global economy becomes increasingly automated and artificial agents become increasingly capable, more research is needed to prepare for novel possibilities that the open-source condition might enable.  Some results in this area might also be applicable to self-fulfilling and self-defeating prophecies in the beliefs and policies of human institutions.

Theorems \ref{thm:cupod}, \ref{thm:dupoc}, 
\ref{thm:pdupoc}, 
\ref{thm:cimcic}, and 
\ref{thm:dimcid} each illustrate how one-shot interactions between open-source agents can sometimes ``short-circuit'' unbounded recursions of metacognition, 
leading to somewhat counterintuitive instances of both mutual cooperation and mutual defection.  Theorem \ref{thm:pdupoc} in particular illustrates how similar results hold for probabilistic open-source agents.  
The properties of the \dupoc agents, and correspondingly their probabilistic variants \pdupoc, are particularly interesting: \dupoc is unexploitable, and for large $k$, \dupoc and \pdupoc reward both cooperation and legibility in their opponents; they even reward the principle of \emph{rewarding} legible cooperation. 
If Open Problem \ref{prob:lob} is answered affirmatively, it might shed light on ways to implement DUPOC-like human institutions.   If Open Problem \ref{prob:holdupoc} is answered affirmatively, it might allow a practical implementation standard for DUPOC-like automated systems.  And, many scenarios that would enable enough mutual understanding to achieve a Nash equilibrium in one round might also allow enough mutual understanding to make and reveal new source code or commitments, allowing better-than-Nash outcomes.

Further open problems need to be resolved before these one-shot open-source interactions can be thoroughly understood.  
Certainly it would be satisfying to implement and investigate an efficient version of \dupoc or \cupod using HOL/ML or Coq as in Open Problem
\ref{prob:holdupoc}, or to look at population dynamics among agents with known interaction outcomes as in Open Problem \ref{prob:modalpopulation}.
In Open Problems 
\ref{prob:dupoccupod}, 
\ref{prob:cupodcimcic}, and 
\ref{prob:dupocdimcid}, there seems to be no obvious self-fulfilling prophecy that can be proven using \pbl.  
For these cases we may need new techniques, or perhaps more general theoretical results such as those described in Open Problems 
\ref{prob:lob}, 
\ref{prob:bgl}, and 
\ref{prob:haskellupgrade}.
As for the successful design of additional resource-bounded proof-based agents, such as the PrudentBot and CDEBot described in Open Problems  
\ref{prob:prudentbot} and \ref{prob:cdebot}, these designs might be easily achievable with currently available methods.  Conversely, attempts at their design might unfold into needing or generating progress on the more general theoretical results.

\bibliography{main}

\appendix

\section{Appendix functions}

\begin{minted}{python}
## file appendix.py
import string

# string_generator iterates lexicographically through all strings
# up to a given length:
def string_generator(length_bound):
    char_set = string.printable
    array = [0]
    char_pos = 1 # end of the string
    while True:
        if array[-char_pos] == len(char_set):
            if char_pos == length_bound:
                break
            for i in range(1,char_pos+1):
                array[-i] = 0
            char_pos += 1
            if char_pos > len(array):
                array = [0]+array
            else:
                array[-char_pos] += 1
            continue
        yield ''.join(char_set[i] for i in array)
        char_pos = 1
        array[-char_pos] += 1
        continue
\end{minted}

\section{Proof system assumptions}\label{app:proofsystem}

We have assumed that the proof system $S$ employed by \Verb{proof_checker} has the following properties typical of real-world proof systems, as discussed by \citet[\S 2.2]{critch2019parametric}:
\begin{enumerate}[label=\alph*)]
    \item $S$ can write down expressions that represent arbitrary computable functions.
    \item $S$ can write down a number $k$ using $\mathcal{O}(\lg(k))$ characters.
    \item $S$ allows for the definition and expansion of abbreviations in the middle of proofs.
    \item \label{item:proofexp} There exists a constant $e^*$ with the following property: Suppose we are given a proof $\rho$ that is $k$ characters long.  Then, it is possible to write out another proof $E(\rho)$ of length at most $e^*\cdot k$, that checks the steps of $\rho$ and verifies that $\rho$ is a valid proof \citep[\S 4.2]{critch2019parametric}.  We call this number $e^*$ a ``proof expansion constant.''  In the notation and terminology of \citet[Definition 4.1]{critch2019parametric}, it defines a ``proof expansion function'', $e(k) = e^*\cdot k$.
\end{enumerate}

Note that \citet{critch2019parametric} operates under the more general assumption that $e(k)$ can be any increasing computable function of $k$.  However, since proofs written in realistic formal proof systems can be checked in linear time \citep{critch2019parametric}, we focus in this paper on the simpler special case where $e(k)$ can be taken to be a linear function $e^*\cdot k$.

    

\auxdef{probcount}{\numberstring{problem}}

\end{document}